\documentclass[conference]{IEEEtran}
\usepackage{subfigure}
\usepackage{setspace}
\usepackage{multirow}
\usepackage{amsmath}
\usepackage{amsthm}
\usepackage{amssymb}
\usepackage{amsfonts}
\usepackage{amscd}
\usepackage{mathrsfs}
\usepackage[final]{graphicx}
\usepackage{graphicx}
\usepackage{psfrag}
\usepackage{color}
\usepackage{url}

\newtheorem{theorem}{Theorem}

\newtheorem{definition}{Definition}

\newtheorem{proposition}{Proposition}
\newtheorem{corollary}{Corollary}
\newtheorem{example}{Example}

\begin{document}

\title{A Singleton Bound for Generalized Ferrers Diagram Rank Metric Codes}

\author{Srikanth B. Pai and B. Sundar Rajan\\
Dept. of ECE, Indian Institute of Science,\\
Bangalore 560012, India\\
Email:\{spai,bsrajan\}@ece.iisc.ernet.in}
\maketitle

\begin{abstract}
In this paper, we will employ the technique used in the proof of classical Singleton bound to derive upper bounds for rank metric codes and Ferrers diagram rank metric codes. These upper bounds yield the rank distance Singleton bound and an upper bound presented by Etzion and Silberstein respectively. Also we introduce {\em generalized Ferrers diagram rank metric code} which is a Ferrers diagram rank metric code where the underlying rank metric code is not necessarily linear. A new Singleton bound for generalized Ferrers diagram rank metric code is obtained using our technique.

\emph{Key words and phrases}: Classical error correcting codes, rank metric codes, subspace codes, Ferrers diagrams, Singleton bound.
\end{abstract}

\section{Introduction}

Singleton Bound for classical error correcting codes was introduced in \cite{Sing}. Since then the proof technique has been carried over to various other forms of coding. A Singleton bound for rank metric codes appears in \cite{Roth} and \cite{Gab} and a quantum Singleton bound appears in \cite{KniLaf}. With the advent of random network coding, a model for error correction was introduced in \cite{KoeKschi}. In this model, subspaces are transmitted and received over a network. A collection of subspaces used for transmission is called a {\em subspace code}. A Singleton bound for constant dimension subspace codes is also derived in \cite{KoeKschi}. Codes that achieve the equality of Singleton bound for classical error correcting codes are called maximum distance separable (MDS) codes and quantum codes that achieve the equality of quantum Singleton bound are called quantum MDS codes. Similarly, codes that achieve the equality of the Singleton bound inequality of rank metric codes are called maximum rank distance (MRD) codes. It is known from \cite{EtzVar} that the equality of the Singleton bound inequality for constant dimension subspace codes cannot be achieved. A partial generalization of the Singleton bound technique was carried out in \cite{SriBSR}. A lattice scheme is defined as a subset of a partially ordered lattice. A Singleton bound for schemes in lattices was proposed in that paper. It was shown that Singleton bound for classical binary codes and subspace codes are special cases of Singleton bound for lattices. 

A Ferrers diagram rank metric code is a subspace code used for the purposes of error correction in RNC.  In \cite{EtzSilb2009}, the authors Etzion and Silberstein introduced the idea of constructing subspace codes via rank metric codes and a combinatorial object called {\em Ferrers diagrams}. They construct subspace codes using this technique and they call these codes {\em Ferrers Diagram Rank Metric Codes}. Their construction procedure involves two steps. In the first step, they choose a classical binary code with minimum distance not less than $d$ and in the second step, they choose a MRD code of a required dimension. We call such a rank metric code as the {\em underlying rank metric code} in this paper. Then they combine these two to form a constant dimension subspace code. They also specify a way to modify these constructions to obtain non-constant dimension codes. They prove an upper bound on the size of a Ferrers diagram rank metric code as well.

In \cite{SriBSR}, Singleton bound for rank metric codes and quantum codes have not been shown to be special cases of Singleton bound for lattices. In this paper, we will demonstrate a technique for proving Singleton bound that proves Singleton bounds of classical algebraic codes, rank metric codes and the upper bound presented in \cite{EtzSilb2009} for Ferrers diagram rank metric codes. The common technique involves deleting dimensions (co-ordinates, rows, columns e.t.c) of codewords until all codewords remain distinct. And then counting the total number of codewords in the deleted space. This technique produces non-linear versions of Singleton bounds for various cases. Non-linear versions of Singleton bounds for classical algebraic codes and rank metric codes already exist in the literature, \cite{Sing} and \cite{Gab} respectively. We prove a non-linear version of the Singleton bound for Ferrers diagram rank metric code in this paper.

The contributions of this paper are as follows:
\begin{enumerate}
\item Non-linear versions of the Singleton bounds for classical algebraic codes, rank metric codes and Ferrers diagram rank metric codes can be derived by a common technique.
\item We generalize Ferrers diagram rank metric codes to include codes where the underlying rank metric code is non-linear.
\item The non-linear version of Singleton bound for Ferrers diagram rank metric codes is a new upper bound. 
\end{enumerate}

The paper is organized as follows: Section \ref{sec_Bac} introduces generalized rank metric codes and generalized Ferrers diagram rank metric codes. Short introductions to Ferrers diagram and row reduced echelon forms are also given in the same section. Section \ref{sec_SB} proves the non-linear versions of the Singleton bounds for classical error correcting codes, rank metric codes and Ferrers diagram rank metric codes using a common technique. Finally, in Section \ref{sec_conc}, we conclude by summarizing our contributions and discussing the scope for future work.

{\it Notations:} A set is denoted by a capital letter and its elements will be denoted by small letters (For example, $ x \in X $). All the sets considered in this paper will be finite. Given a set $X$, $|X|$ denotes the number of elements in the set. $\mathbb{F}_q$ denotes the finite field with $q$ elements where $q$ is a power of a prime number. $\mathbb{F}_q^n$ denotes the $n$ dimensional vector space of $n$-tuples over $\mathbb{F}_q$. $M_{m \times n}(q)$ denotes the set of all $m \times n$ matrices with entries from $\mathbb{F}_q$. $\lfloor x\rfloor$ stands for the greatest integer function of $x$.

\section{Background}
\label{sec_Bac}
Instead of saying {\em non-linear versions of Singleton bounds}, for the rest of the paper we will derive Singleton bounds for generalized versions of linear codes. In other words, a code is linear by default. When the adjective {\em generalized} is attached to a code, it means that the code need not be linear. A generalized classical $q$-nary code is a subset of $\mathbb{F}_q^n$ and a classical $q$-nary code is a subspace of $\mathbb{F}_q^n$. In this section, we introduce generalized rank metric codes and generalized Ferrers diagram rank metric codes.

\subsection{Generalized Rank Metric Codes}

A {\em generalized rank metric code} $C$ is defined as a subset of $M_{m \times n}(q)$ equipped with a metric, called the \emph{Rank distance}, $d_R$. Given two elements $A, B \in C$, \[d_R(A,B) := \text{rank}(A-B).\] An element of a generalized rank metric code is called its \emph{codeword}. The \emph{minimum distance} of a generalized rank metric code is defined as the minimum possible distance between two different codewords of the code. If $C$ is a subspace of $M_{m \times n}(q)$, then such a code is called a \emph{rank metric code}. 

In the literature, rank metric codes are also called ``linear array codes'' \cite{Roth}. The following observation captures the connection between error correction capability of a rank metric code and its minimum distance.

\begin{proposition}\cite[Section 1]{Roth}
If $d$ is the minimum distance of a generalized rank metric code $C$, then $C$ can correct $t=\lfloor \frac{d-1}{2} \rfloor$ or fewer errors, and conversely.
\end{proposition}

A rank metric code $C \subseteq M_{m \times n}(q)$ with minimum distance $d$ is called a $[m \times n, \dim(C), d]$ code. We shall call a generalized rank metric code $C \subseteq M_{m \times n}(q)$ with minimum distance $d$ as a $(m \times n, |C|, d)$ code. 

The following example constructs a simple rank metric code.

\begin{example}
Consider the following generalized rank metric code $C$,
\begin{align*}C = \left\{ \left( \begin{array}{cc}
	1 & 1 \\
	1 & 1
\end{array} \right), \left( \begin{array}{cc}
	1 & 0 \\
	0 & 1
\end{array} \right) \right\}.
\end{align*}
The rank distance between the two codewords is two units and $C$ is subset of $M_{2 \times 2}(2)$. Therefore $C$ is an example of a $(2 \times 2, 2, 2)$ generalized rank metric code. It is not a rank metric code since $C$ is not a subspace of $M_{2 \times 2}(2)$.
\end{example}

Next we present a rank metric code constructed by Gabidulin in \cite{Gab}.

\begin{example}
Let $\alpha$ be a root of the irreducible polynomial $f(x) = x^3 + x + 1$, where $f(X) \in \mathbb{F}_2[X]$. The extension field ${\mathbb{F}_{2^3}}$ is the smallest field that contains both $\mathbb{F}_2$ and the element $\alpha$. The subset $\{1,\alpha,\alpha^2\}$ is linearly independent over the base field $\mathbb{F}_2$. Now define a code $C$ as the set of all vectors $x \in {\mathbb{F}_{2^3}}^3$ such that $xH^T = 0$ where 
\begin{align*}
H =  \left( \begin{array}{ccc}
	1 & \alpha ^2	& \alpha
\end{array} \right).
\end{align*}
Therefore $C$ is the subspace spanned by the linearly independent set of column vectors $\{(1,\alpha,1), (0,1,\alpha)\}$. Now if we pick a basis $\{1, \alpha, \alpha^2\}$ for the extension field ${\mathbb{F}_{2^3}}$ over $\mathbb{F}_2$, we can write each co-ordinate of every vector in $C$ as a vector from ${\mathbb{F}_{2}}^3$. For example,

\begin{align*}
\left[\begin{array}{c} 1 \\ \alpha \\ 1 \end{array}\right] = \left( \begin{array}{ccc}
	1 & 0 & 0 \\
	0 & 1 & 0 \\
	1 & 0 & 0 
\end{array} \right),
\end{align*} 

\begin{align*}
\left[\begin{array}{c} 1 \\ \alpha + 1 \\ \alpha + 1 \end{array}\right] = \left( \begin{array}{ccc}
	1 & 0 & 0 \\
	1 & 1 & 0 \\
	1 & 1 & 0
\end{array} \right).
\end{align*} 

It can be verified that we have a $[3 \times 3, 2, 2]$ rank metric code. We note that $n-k = d-1$ for this example where $n=3$, $k=2$, and $d=2$. This is an example of an MRD code. A general construction is presented in \cite{Gab}. 
\end{example}

\subsection{Ferrers diagrams and Reduced Row Echelon Forms} 


Ferrers diagrams are representations of partitions of natural numbers. They are used in combinatorics as a tool to derive certain results about recursions, and generating functions that relate to partitions. In this section, we will introduce Ferrers diagrams and its connections to row reduced echelon forms of matrices. Before we introduce Ferrers diagrams, we need the definition of a \emph{partition} of a natural number.

\begin{definition}
A \emph{partition} of $n$ is a representation of $n$ as an unordered sum of positive integers. Each summand in the partition is called a \emph{part}.
\end{definition}

\begin{example}
The representation $5=1+1+3$ is an example of a partition of for the natural number $5$. Note that $5=1+3+1$ is not considered as a different partition of $5$ because both representations contain the same positive integers in a different order. $3,1,$ and $1$ are parts of the partition $5=1+1+3$.
\end{example}

Traditionally the parts of a partition are listed in a decreasing order. For example, the partitions of $4$ are represented as:
\[4=3+1=2+2=2+1+1=1+1+1+1.\]

\begin{figure}
 \centering
 \includegraphics[scale=0.4, trim=90 120 50 10, clip=true]{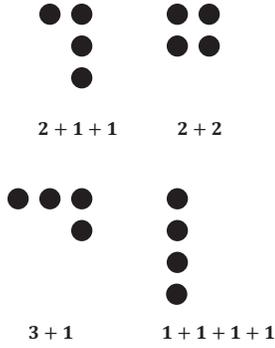}
 \caption{All the Ferrers diagrams of number four.}
 \label{fig_four_FD}
\end{figure}

The following formal definition of Ferrers diagram is explained below, adapted from \cite{EtzSilb2009}:
\begin{definition}
A \emph{Ferrers diagram} is a pattern of dots with the $i$-th row having the same number of dots as the $i$th
term in the partition. A Ferrers diagram satisfies the following conditions:
\begin{enumerate}
\item The number of dots in a row is less than or equal to the number of
dots in the previous row.
\item All the dots are aligned to the right of the diagram.
\end{enumerate}
\end{definition}

For example, the different Ferrers diagrams of the number $4$ has been shown in Fig. \ref{fig_four_FD}. The number of rows and columns of a Ferrers diagram $F$ is the maximum number of dots among all columns and all rows of $F$ respectively. An $m \times η$ Ferrers diagram is a Ferrers diagram with $m$ rows and $n$ columns. If we think of the Ferrers diagram $F$ as a matrix and transpose the diagram across the secondary diagonal, we get a diagram that is called the conjugate of $F$. Note that it may represent another partition. Due to the nature of a matrix transpose, a $m \times η$ Ferrers diagram gives rise to a conjugate that is a $η \times m$ Ferrers diagram. An example of a Ferrers diagram and its conjugate is shown in Fig. \ref{fig_conjFerr}.

\begin{figure}
 \centering
 \includegraphics[scale=0.4, trim=90 120 50 10, clip=true]{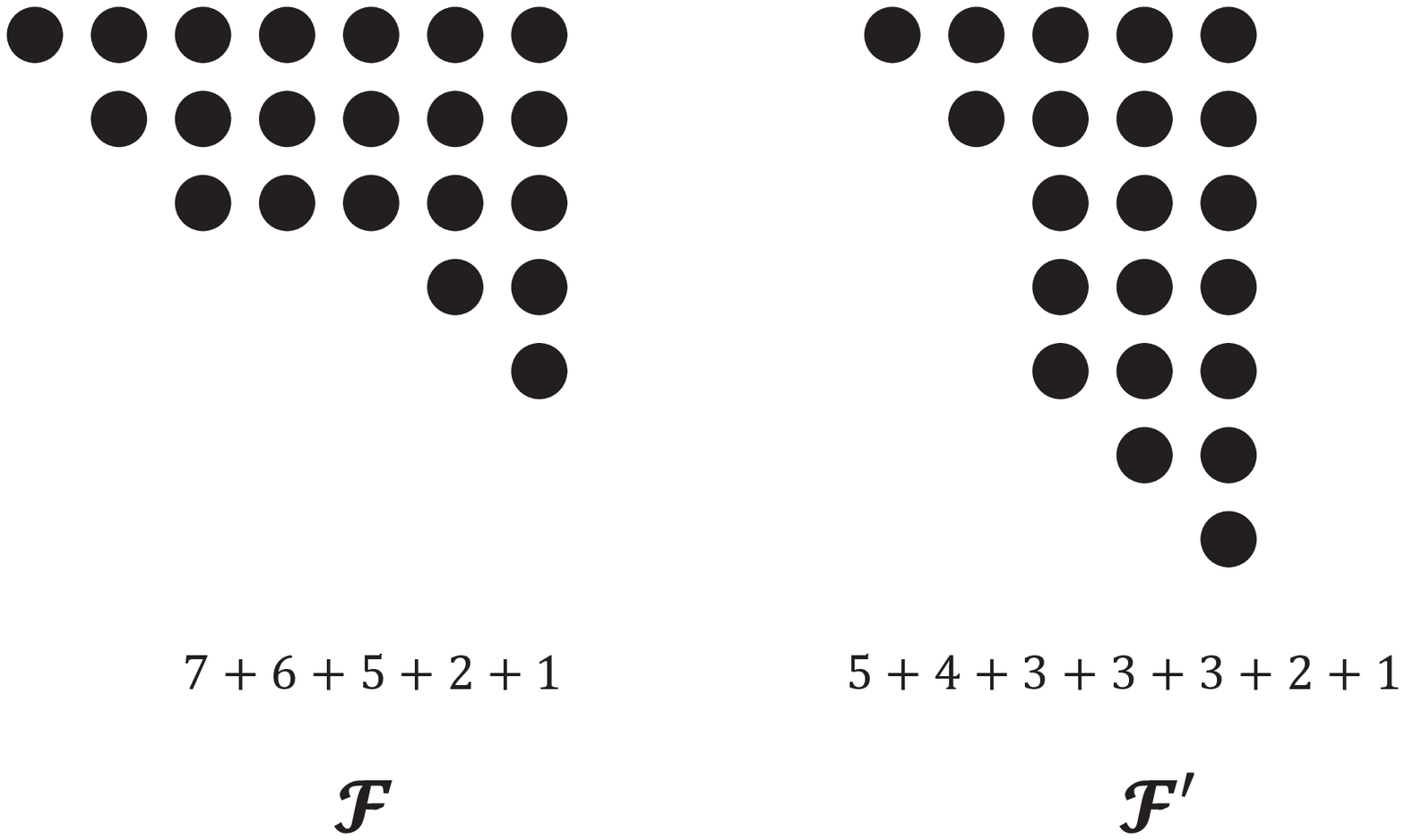}
 \caption{${\cal F}'$ is the conjugate of $\cal F$.}
 \label{fig_conjFerr}
\end{figure}

Now we will define the notion of a row reduced echelon form for a matrix from \cite{EtzSilb2009}. The notion of a row reduced echelon form is defined as follows:

\begin{definition}
A matrix is in row reduced echelon form (RREF) if
\begin{enumerate}
    \item The first nonzero number from the left of a nonzero row, called the leading coefficient of that row, is always strictly to the right of the leading coefficient of the row above it. 
		\item The leading coefficient of every row is always one.
		\item The leading coefficient is the only nonzero entry in its column.
		\item Row of all-zeroes is not written at all.
\end{enumerate}
\end{definition}

The last condition in the above definition is not standard. We note that if a RREF matrix is a $k \times n$ matrix, then the rank of the matrix is $k$. 

Subspace codes are described by listing subspaces and one possible way to list subspaces is by listing a basis of those subspaces. If we construct a basis for a given subspace, we can use Gaussian elimination, a sequence of elementary transformations, to make it a RREF matrix \cite{LintWils}. In other words, every subspace of $\mathbb{F}_q^n$ is a row-space of some RREF matrix. It is well known that a RREF is unique for a given subspace of $\mathbb{F}_q^n$. Therefore, given a subspace $W$ of $\mathbb{F}_q^n$, let $\text{R}(W)$ denote the RREF of $W$. A RREF matrix without the condition that the leading coefficient of each row is $1$ is simply called a \emph{row echelon form}. RREF of matrices are often used in linear algebra to solve a system of linear equations. The following example constructs the RREF of a $3$-dimensional subspace in a $7$ dimensional binary space.

\begin{example}
\label{ex_rref}
Consider the subspace $W$ of $\mathbb{F}_2^7$ consisting of the following vectors:
\begin{align*}
(1 \quad 0 \quad 1 \quad 0 \quad 0 \quad 1 \quad 1), \\
(1 \quad 0 \quad 0 \quad 0 \quad 1 \quad 1 \quad 0), \\
(0 \quad 0 \quad 0 \quad 0 \quad 1 \quad 1 \quad 1), \\
(0 \quad 0 \quad 1 \quad 0 \quad 1 \quad 0 \quad 1), \\
(1 \quad 0 \quad 0 \quad 0 \quad 0 \quad 0 \quad 1), \\
(1 \quad 0 \quad 1 \quad 0 \quad 1 \quad 0 \quad 0), \\
(0 \quad 0 \quad 1 \quad 0 \quad 0 \quad 1 \quad 0), \\
(0 \quad 0 \quad 0 \quad 0 \quad 0 \quad 0 \quad 0). 
\end{align*}

Clearly, $X$ is a three dimensional binary space. We can pick a basis by picking three linearly independent vectors. The first row can be picked in $7$ ways. The second can be picked in $6$ ways. Now we cannot pick the sum of two rows as the third row. Therefore, we have $4$ ways of picking the last row. Note that the choices of each row are independent of the other rows. Therefore, there are $168$ different $3 \times 7$ matrices whose row-space is $X$. But the unique RREF of $X$, among $168$ different matrices whose rows span $X$, is given by the following matrix:

\begin{align*}\text{R}(X) = \left( \begin{array}{ccccccc}
	1 & 0 & 0 & 0 & 0 & 0 & 1 \\
	0 & 0 & 1 & 0 & 0 & 1 & 0 \\
	0 & 0 & 0 & 0 & 1 & 1 & 1
\end{array} \right)
\end{align*}
\end{example}

For every $k$-dimensional subspace $X$ of $\mathbb{F}_q^n$, we can define a $n$-length binary vector $v(X)$, called the \emph{identifying vector of $X$} in \cite{EtzSilb2009}, where the ones in $v(X)$ are in the positions (columns) where $\text{R}(X)$ has the leading ones. 

\begin{example}
For the $X$ introduced in Example \ref{ex_rref}, $v(X)$ is $(1,0,1,0,1,0,0)$. 
\end{example}

The identifying vector of a $k \times n$ RREF matrix is a binary vector of Hamming weight $k$. Also, note that there can be multiple subspaces which have the same identifying vector. Given a binary vector $v$ of length $n$ and weight $k$, the echelon Ferrers form of $v$ is denoted by $EF(v)$ is the $k \times n$ RREF matrix with leading entries (of rows) in the columns indexed by the nonzero entries of $v$ and `` $\bullet$ '' in all entries which do not have terminals zeroes or ones. A `` $\bullet$ '' is termed \emph{a dot}. This notation is also given in \cite{LintWils}. The dots of this matrix form the Ferrers diagram of $EF(v)$. If we substitute elements of $\mathbb{F}_q$ in the dots of $EF(v)$ we obtain a $k$-dimensional subspace $X$ of $P_q(n)$. The form $EF(v)$ will be called also the echelon Ferrers form of $X$.

\begin{example}
Let $v = (1,1,0,0,1,0,0)$, then the echelon Ferrers form $EF(v)$ is a $3 \times 7$ matrix:
\begin{align*}\text{R}(X) = \left[ \begin{array}{ccccccc}
	1 & 0 & \bullet & \bullet & 0 & \bullet & \bullet \\
	0 & 1 & \bullet & \bullet & 0 & \bullet & \bullet \\
	0 & 0 & 0 & 0 & 1 & \bullet & \bullet
\end{array} \right]
\end{align*}

The Ferrers diagram associated with $EF(v)$ is given by the following $3 \times 4$ array:

\begin{align*}
{\cal F} = \begin{array}{cccc}
	\bullet & \bullet & \bullet & \bullet \\
  \bullet & \bullet & \bullet & \bullet \\
	  &   & \bullet & \bullet
\end{array}
\end{align*}
\end{example}

Given a subspace $X$, when the Ferrers diagram associated with $EF(v(X))$ is filled with the entries in the original matrix corresponding to the entries in the locations of the $\bullet$, we say that that it is the \emph{matrix associated with Ferrers diagram of $R(X)$}.
%

\begin{example}
Consider the RREF of a subspace $X$,
\begin{align*}\text{R}(X) = \left( \begin{array}{ccccccc}
	1 & 0 & 0 & 0 & 0 & 0 & 1 \\
	0 & 0 & 1 & 0 & 0 & 1 & 0 \\
	0 & 0 & 0 & 0 & 1 & 1 & 1
\end{array} \right)
\end{align*}

The identifying vector of $R(X)$ is $(1,0,1,0,1,0,0)$ and the associated Ferrers diagram is 
\begin{align*}
{\cal F} = \begin{array}{cccc}
	\bullet & \bullet & \bullet & \bullet \\
          & \bullet & \bullet & \bullet \\
	              &   & \bullet & \bullet
\end{array}
\end{align*}

The matrix associated with the Ferrers diagram of $R(X)$ is 
\begin{align*}
A = \begin{array}{cccc}
	0 & 0 & 0 & 1 \\
    & 0 & 1 & 0 \\
	  &   & 1 & 1
\end{array}
\end{align*}

\end{example}

Now we are in a position to define a generalized Ferrers diagram rank metric code. Given a binary vector $v$ of length $n$ and weight $k$, let $EF(v)$ be the echelon Ferrers form of $v$. Let $F$ be the Ferrers diagram of $EF(v)$. Then $F$ is an $a \times b$ Ferrers diagram, where $a \leq k, b \leq n - k$. A code $C$ is an $[F, p̺, d]$ {\em Ferrers diagram rank-metric code} if all codewords are $a \times b$ matrices in which all entries not in $F$ are zeroes and $C$ is also a $[a \times b, p,d]$ rank metric code. We call the rank metric code as the rank metric code associated with the Ferrers diagram rank-metric code. A \emph{generalized Ferrers diagram rank metric code} is one where the associated array code need not be linear. We represent such a code as $(F, M, d)$ code where $M$ denotes the number of codewords, $F$ is a Ferrers diagram and $d$ is the minimum rank distance. 


\section{Singleton Bounds}
\label{sec_SB}
In this section, we will present the Singleton for generalized classical $q$-nary codes and use the proof technique for generalized rank metric codes and generalized Ferrers diagram rank metric codes. 

\subsection{Classical Singleton Bound}
The proofs of rest of the Singleton bounds is based on the proof of the classical Singleton bound originally given in \cite{Sing}. The following theorem and proof is well known and is available in introductory text books. It is repeated here for the sake of completion and the proof technique will be used repeatedly subsequently. 

\begin{theorem}\cite{Sing}
If $C$ is binary code in $\mathbb{F}_q^n$ with minimum distance $d$, then \[|C| \leq q^{n-d+1}.\]
\end{theorem}
\begin{proof}
We modify the code by deleting $d-1$ co-ordinates from all the codewords. Since the minimum Hamming distance was $d$, if two such modified codewords are now equal, then they differed in at most $d-1$ places. This would imply that the minimum Hamming distance is strictly less than $d$ which is a contradiction. Therefore $|C|$ remains unchanged even after deletion. But the modified code has only $n-d+1$ co-ordinates. And the maximum number of $q$-nary vectors of length $n-d+1$ is $q^{n-d+1}$ and thus $|C| \leq q^{n-d+1}.$
\end{proof}


\subsection{Rank Metric Singleton Bound}

We shall now derive the Singleton bound for generalized rank metric codes.

\begin{theorem}\cite{Gab}
\label{thm_rdSB}
If $C$ is a code in $M_{m \times n}(q)$ and the minimum rank distance of the code is $d$, then \[|C| \leq q^{\text{min}\{m(n-d+1),n(m-d+1)\}}.\]
\end{theorem}
\begin{proof}
We delete $d-1$ columns from the all the codewords in $C$ (See Figure \ref{fig_rk_SB}). Now, if some two codewords are equal then these two codewords differed in at most $d-1$ columns before deletion. But this would mean that the rank distance between the two codewords originally was less than $d$. This contradiction shows that all the codewords in the deleted code are distinct. We can similarly prove that the codewords are distinct if we punctured $d-1$ rows. So the deleted codewords belong to the set $M_{m \times (n-d+1)}(\mathbb{F}_q)$ or $M_{(m - d + 1) \times n}(\mathbb{F}_q).$ Therefore \[|C| \leq q^{\text{min}\{m(n-d+1),n(m-d+1)\}}.\]
\end{proof}

\begin{figure}
 \centering
 \includegraphics[scale=0.4, trim=90 120 50 10, clip=true]{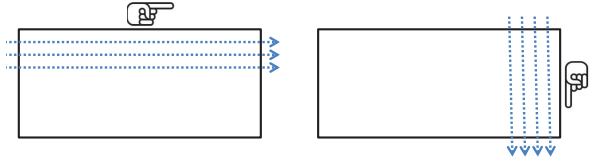}
 \caption{Strategy for generalized Rank Metric Singleton bound: Puncturing $d-1$ rows or columns does not change the size of the code. Here $d$ is the minimum distance of the code $C$.}
 \label{fig_rk_SB}
\end{figure}

Our proof has the distinction that it does not need the assumption of linearity and therefore our rank metric Singleton bound is true for generalized rank metric codes. It should be noted that the proof in \cite{Gab} is simpler than our proof. In particular, we have the following corollary for a $[m \times n, k, d]$ rank metric code.

\begin{corollary}\cite{Roth}
If $C$ is a $[m \times n, k, d]$ rank code in $M_{m \times n}(q)$, then 
\[k \leq \text{min}\{m(n-d+1),n(m-d+1)\}.\]
\end{corollary}
\begin{proof}
A rank metric code of dimension $k$ has $q^k$ codewords. Applying Theorem \ref{thm_rdSB} to $C$ proves that \[q^k \leq q^{\text{min}\{m(n-d+1),n(m-d+1)\}}.\]
Therefore, we have \[k \leq \text{min}\{m(n-d+1),n(m-d+1)\}.\]
\end{proof}

\subsection{Singleton Bound for Generalized Ferrers Diagram Rank Metric Codes}

We will prove a Singleton bound for generalized Ferrers diagram rank metric code and obtain the Singleton bound for Ferrers diagram rank metric code codes, as a corollary, by specializing it for rank metric codes. The proof proceeds in a manner very similar to our proof for classical Singleton bound and Singleton bound for generalized rank metric codes. It is important to note that the proof provided in \cite{EtzSilb2009} uses the subspace structure of rank metric codes and is therefore applicable only to Ferrers diagram rank metric codes. Our proof, on the other hand, is more general since we do not use the subspace structure of the code.

\begin{figure}
 \centering
 \includegraphics[scale=0.4, trim=90 120 50 10, clip=true]{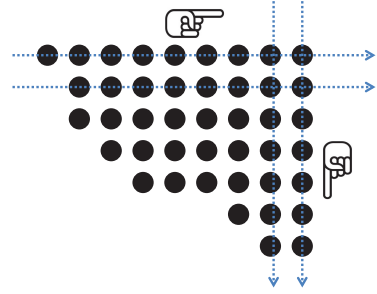}
 \caption{Deleting rows and columns in a general Ferrers diagram rank metric code.}
 \label{fig_EtzSilbSB}
\end{figure}

\begin{theorem}
\label{thm_genEtzSilbSB}
For a given Ferrers diagram $F$, and non-negative integers $d$ and $i, 0 \leq i \leq d-1$, if $v_i$ is the number of dots in $F$, which are not contained in the first $i$ rows and are not contained in the rightmost $d - 1 - i$ columns and ${\cal C}$ is a $(F, M, d)$ general Ferrers diagram rank metric code, then
\[\displaystyle M \leq q^{\min_i{v_i}}\].
\end{theorem}
\begin{proof}
It suffices to prove that for all $i$ between $0$ and $d-1$: \[ M \leq q^{{v_i}}. \]

As shown in Fig. \ref{fig_EtzSilbSB} when we delete (or puncture) the first $i$ rows and the last $d - 1 - i$ columns of all the codewords in the $(F, M, d)$ general Ferrers diagram rank metric code, we shall call this a \emph{puncturing} of the general Ferrers diagram rank metric code (or simply termed `puncturing' for the reminder of the proof). Assume that $F$ is a Ferrers diagram of dimension $a \times b$.

If some two codewords $X$ and $Y$ become equal after puncturing, then those two codewords differed only in the first $i$ rows or the last $d-1-i$ columns. Fig. \ref{fig_rankcal} shows the matrix structure of the difference between $X$ and $Y$ where $A$,$B$ and $C$ are submatrices, and $O$ is an all zero matrix of appropriate dimensions. $A$ is an $i \times (b-d+i+1)$ matrix, $B$ is an $i \times (d-1-i)$ matrix and $C$ is a $(a-i) \times (d-1-i)$ matrix. The rank of $X-Y$ is at most the sum of rank of $A$ and rank of $C$. However, rank of $A$ is at most $i$ and rank of $C$ is at most $d-1-i$, therefore the rank of $X-Y$ is at most $i + (d-1-i) = d-1.$ This means that the rank distance between $X$ and $Y$ is strictly less than $d$ which contradicts the fact that the minimum rank distance of $\cal C$ is $d$. Therefore, two codewords $X$ and $Y$ cannot be equal after puncturing the first $i$ rows and the last $d - 1 - i$ columns of all the codewords. This implies that the number of codewords have not changed after the process of deletion. There are only $v_i$ locations in the matrix where two codewords can differ after puncturing, from which we can obtain at most $q^{v_i}$ matrices and thus \[ M \leq q^{v_i}.\] But this is true for every $i, 0 \leq i \leq d-1$ which proves the theorem.
\end{proof}

\begin{figure}
 \centering
 \includegraphics[scale=0.4, trim=60 50 50 10, clip=true]{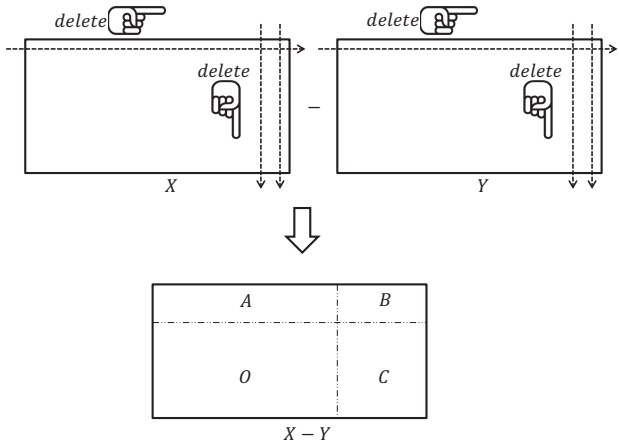}
 \caption{The structure of the difference to two codewords which are equal after deletion of rows and columns.}
 \label{fig_rankcal}
\end{figure}

Now we state and prove the Singleton bound for Ferrers diagram rank metric codes \cite{EtzSilb2009}[Thm. 1].
\begin{corollary}
\label{cor_EtzSilbSB}
For a given $i, 0 \leq i \leq d-1$, if $v_i$ is the number of dots in $F$, which are not contained in the first $i$ rows and are not contained in the rightmost $d - 1 - i$ columns then
$\displaystyle \min_i{v_i}$ is an upper bound of $\dim(F, d)$.
\end{corollary}
\begin{proof}
The the number of codewords in a $[F, p, d]$ Ferrers diagram rank metric code is $q^p$. Therefore, applying Theorem \ref{thm_genEtzSilbSB} to the given code we obtain
\[\displaystyle q^p = M \leq q^{\min_i{v_i}}.\]
Therefore \[p \leq \min_i{v_i}.\]
\end{proof}

\section{Conclusion}
\label{sec_conc}
We used the proof technique employed in \cite{Sing} to obtain the non-linear versions of Singleton bounds for rank metric codes and Ferrers diagram rank metric codes. The non-linear version of Singleton bound for Ferrers diagram rank metric code is new.
It is not clear if this proof technique can be used to derive the quantum Singleton bound. Investigating the connections between this proof technique and Singleton bound for lattice schemes is also an interesting direction that can be pursued.

\end{document}